\documentclass[USenglish, 11pt]{article}
\usepackage[T1]{fontenc}
\usepackage{amsthm,amsmath,amssymb,amsfonts,authblk,hyperref,setspace}
\usepackage{graphicx}
\usepackage{color}
\usepackage{authblk}
\usepackage{cite}
\usepackage{stfloats}
\usepackage{url}

\usepackage{float}
\usepackage{tikz}
\newtheorem{theorem}{Theorem}[section]

\newtheorem{lemma}[theorem]{Lemma}
\newtheorem{claim}[theorem]{Claim}
\newtheorem{proposition}[theorem]{Proposition}

\usepackage{fullpage}
   
\usepackage[]{algorithm2e}
\SetKwInOut{Input}{Input}
\SetKwInOut{Output}{Output}

\usepackage{float}
\usepackage{tikz}
\newcommand{\comment}[1]{}
\newcommand{\tildeO}{\widetilde{O}}
\newcommand{\tO}{\widetilde{O}}

\newcommand{\N}{\mathbb{N} \cup \{0\}}

\newcommand{\kt}{\tilde{k}}
\newcommand{\editd}{\mathbf{d}_\textrm{edit}}

\newcommand{\ed}{\Delta_\textrm{edit}}
\newcommand{\cost}{\textrm{cost}}

\newcommand{\source}{\textrm{source}}
\newcommand{\sink}{\textrm{sink}}

\newcommand{\poly}{\textrm{poly}}

\title{Approximate Online Pattern Matching in Sub-linear Time\footnote{The research leading to these results has received funding from the European Research Council under the European Union's Seventh Framework Programme (FP/2007-2013)/ERC Grant Agreement no. 616787.}}
\author[1]{Diptarka Chakraborty\thanks{diptarka@iuuk.mff.cuni.cz}}
\author[2]{Debarati Das\thanks{debaratix710@gmail.com}}
\author[3]{Michal Kouck{\'{y}}\thanks{koucky@iuuk.mff.cuni.cz}}
\affil[1]{Faculty of Mathematics and Computer Science,
The Weizmann Institute of Science,
Rehovot, Israel}
\affil[2,3]{Computer Science Institute of Charles University,
Malostransk{\'e}  n{\'a}m\v{e}st\'{\i} 25,
118 00 Praha 1, Czech Republic}

\begin{document}

\maketitle

\begin{abstract}
	We consider the approximate pattern matching problem under edit distance. In this problem we are given a pattern $P$ of length $w$ and a text $T$ of length $n$ over some alphabet $\Sigma$, and a positive integer $k$. The goal is to find all the positions $j$ in $T$ such that there is a substring of $T$ ending at $j$ which has edit distance at most $k$ from the pattern $P$. Recall, the edit distance between two strings is the minimum number of character insertions, deletions, and substitutions required to transform one string into the other. For a position $t$ in $\{1,...,n\}$, let $k_t$ be the smallest edit distance between $P$ and any substring of $T$ ending at $t$. In this paper we give a constant factor approximation to the sequence $k_1,k_2,...,k_{n}$. We consider both offline and online settings.

In the offline setting, where both $P$ and $T$ are available, we present an algorithm that for all $t$ in $\{1,...,n\}$, computes the value of $k_t$ approximately within a constant factor. The worst case running time of our algorithm is $O(n w^{3/4})$. As a consequence we break the $O(nw)$-time barrier for this problem.

In the online setting, we are given $P$ and then $T$ arrives one symbol at a time. We design an algorithm that upon arrival of the $t$-th symbol of $T$ computes $k_t$ approximately within $O(1)$-multiplicative factor and $w^{8/9}$-additive error. Our algorithm takes $O(w^{1-(7/54)})$ amortized time per symbol arrival and takes $O(w^{1-(1/54)})$ additional space apart from storing the pattern $P$.

Both of our algorithms are randomized and produce correct answer with high probability. To the best of our knowledge this is the first worst-case sub-linear (in the length of the pattern) time and sub-linear succinct space algorithm for online approximate pattern matching problem. To get our result we build on the technique of Chakraborty, Das, Goldenberg, Kouck\'y and Saks (appeared in FOCS'18) for computing a constant factor 
approximation of edit distance in sub-quadratic time.
\end{abstract}

 \section{Introduction}
Finding the occurrences of a pattern in a larger text is one of the fundamental problems in computer science. Due to its immense applications this problem has been studied extensively under several variations~\cite{KMP77, GS81, Abr87, Cro92, GPR95, CGPR95, Indyk98, Nav01, KP18}. One of the most natural variations is where we are allowed to have a small number of errors while matching the pattern. This problem of pattern matching while allowing errors is known as \emph{approximate pattern matching}. The kind of possible errors varies with the applications. Generally we capture the amount of errors by the distance metric defined over the set of strings. One common and widely used distance measure is the edit distance (aka \emph{Levenshtein distance})~\cite{Lev65}. The edit distance between two strings $T$ and $P$ denoted by $\editd(T,P)$ is the minimum number of character insertions, deletions, and substitutions required to transform one string into the other. In this paper we focus on the approximate pattern matching problem under edit distance. This problem has various applications ranging from computational biology, signal transmission, web searching, text processing to many more.

Given a pattern $P$ of length $w$ and a text $T$ of length $n$ over some alphabet $\Sigma$, and an integer $k$ we want to identify all the substrings of $T$ at edit distance at most $k$ from $P$. As the number of such substrings might be quadratic in $n$ and one wants to obtain efficient algorithms one focuses on finding the set of all right-end positions in $T$ of those substrings at distance at most $k$. More specifically, for a position $t$ in $T$, we let $k_t$ be the smallest edit distance of a substring of $T$ ending at $t$-th position in $T$. (We number positions in $T$ and $P$ from $1$.) The goal is to compute the sequence $k_1,k_1,\dots,k_{n}$ for $P$ and $T$. Using basic dynamic programming paradigm we can solve this problem in $O(nw)$ time~\cite{SEL80}. Later Masek and Paterson~\cite{MP80} shaved a $\log n$ factor from the above running time bound. Despite of a long line of research, this running time remains the best till now. Recently, Backurs and Indyk~\cite{BI15} indicate that this $O(nw)$ bound cannot be improved significantly unless the Strong Exponential Time Hypothesis (SETH) is false. Moreover Abboud et al.~\cite{AHWW16} showed that even shaving an arbitrarily large polylog factor would imply that NEXP does not have non-uniform ${NC}^1$ circuits which is likely but hard to prove conclusion. More hardness results can be found in ~\cite{ABW15,BK15, AB17, AR18}.

In this paper we focus on finding an approximation to the sequence $k_1,k_1,\dots,k_{n}$ for $P$ and $T$. For reals $c,k\ge 0$, a sequence $\kt_1,\dots,\kt_{n}$ is
$(c,k)$-approximation to $k_1,\dots,k_{n}$, if for each $t\in \{1,\dots,n\}$, $k_t \le \kt_t \le c \cdot k_t + k$. 
Hence, $c$ is the multiplicative error and $k$ is the additive error of the approximation.
An algorithm computes $(c,k)$-approximation to approximate pattern matching if it outputs a $(c,k)$-approximation of the true sequence $k_1,k_1,\dots,k_{n}$ for $P$ and $T$. 
We refer $(c,0)$-approximation simply as $c$-approximation. Our main theorem is the following.

% the objective is to compute for each index $j \in [n]$, the value of $\min_{i \le j} \ed(P,T[i,\cdots ,j])$. In this paper we study this problem in 
% online setting. Our setting is as follows: pattern $P$ is given and then the text $T$ comes in online fashion, one character at a time. At any time $t$ 
% we want to compute the value of $\min_{i \le t} \ed(P,T[i,\cdots ,t])$.

\begin{theorem}
\label{thm:main-offline}
There is a constant $c\ge 1$ and there is a randomized algorithm that computes $c$-approximation to approximate pattern matching in time $O(n\cdot w^{3/4})$ with probability at least $(1-1/n^3)$.
\end{theorem}

In the recent past researchers also studied the approximate pattern matching problem in the online setting. The online version of this pattern matching problem mostly arises in real life applications that require matching pattern in a massive data set, like in telecommunications, monitoring Internet traffic, building firewall to block viruses and malware connections and many more. The online approximate pattern matching is as follows: we are given a pattern $P$ first and then the text $T$ is coming symbol by symbol. Upon receipt of the $t$-th symbol we should output the corresponding $k_t$. The online algorithm runs in {\em amortized time} $O(\ell)$ if it runs in total time $O(n \cdot \ell)$ and it uses {\em succinct space} $O(s)$ if in addition to storing $P$ it uses at most $O(s)$ cells of memory at any time.

\begin{theorem}
\label{thm:main-online}
There is a constant $c\ge 1$ so that there is a randomized online algorithm that computes $(c,w^{8/9})$-approximation to approximate pattern matching in amortized time $O(w^{1-(7/54)})$ and succinct space $O(w^{1-(1/54)})$ with probability at least $1-1/\poly(n)$.
\end{theorem}

To the best of our knowledge this is the first online approximation algorithm that takes sublinear (in the length of the pattern) running time and sublinear succinct space for the approximate pattern matching problem. The succinct space data structure is quite natural from the practical point of view and has been considered for many problems including pattern matching, e.g.~\cite{Pat08, GST17}. 

To prove our result we use the technique developed by Chakraborty, Das, Goldenberg, Kouck{\'{y}} and Saks in~\cite{CDGKS18}, where they provide a sub-quadratic time constant factor approximation algorithm for the edit distance problem. Suppose one has only a black-box access to a sub-quadratic time approximation algorithm for computing the edit distance. 
It is not clear how to use that algorithm to design an algorithm for the offline approximate pattern matching problem that runs in $O(nw^{1-\epsilon})$ time, for some $\epsilon > 0$. 
So even given the result of~\cite{CDGKS18} it was still open whether one can solve approximate pattern matching problem in time better than $O(nw)$. 

In this paper we first design an offline algorithm by building upon the technique used in~\cite{CDGKS18}. To do this we exploit the similarity between the "dynamic programming graphs" (see Section~\ref{sec:prel}) for approximate pattern matching problem and the edit distance problem. 
As witnessed for example by the running time of our pattern matching algorithm, which is $O(nw^{3/4})$, whereas the running time of the edit distance algorithm is $O(n^{1+5/7})$, this still requires careful
modifications to the edit distance algorithm.
However the scenario becomes more involved if one wants to design an online algorithm using only a small amount of extra space. The approximation algorithm for edit distance in~\cite{CDGKS18} works in two phases: first a covering algorithm is used to discover a suitable set of shortcuts in the pattern matching graph, and then a min-cost path algorithm on a grid graph with the shortcuts yields the desired result.
In the online setting we carefully interleave all of the above phases. However that by itself is not sufficient since the first phase, i.e., the covering algorithm used in~\cite{CDGKS18} essentially relies on the fact that both of the strings are available at any point of time. We modify the covering technique so that it can also be implemented in the situation when we cannot see the full text. We show that if we store the pattern $P$ then we need only $O(w^{1-\gamma})$ extra space to perform the sampling. Furthermore, the min-cost path algorithm in~\cite{CDGKS18} takes $O(w)$ space. We modify that algorithm too in a way so that it also works using only $O(w^{1-\gamma})$ space. We describe our algorithm in more details in Section~\ref{sec:online-pattern-matching}.

\subsection{Related work}
The approximate pattern matching problem is one of the most extensively studied problems in modern computer science due to its direct applicability to data driven applications. In contrast to the exact pattern matching here a text location has a match if the distance between the pattern and the text is within some tolerated limit. In our work we study the approximate pattern matching under edit distance metric. The very first $O(nw)$-time algorithm was given by Sellers~\cite{SEL80} in $1980$. Masek and Paterson~\cite{MP80} proposed an $O(nw/\log n)$-time $O(n)$-space algorithm using Four Russians~\cite{ADKF75} technique. Later ~\cite{M86,LV89,GP90} gave $O(kn)$-time algorithms where $k$ is the upper limit of allowed edit operations. All of these algorithms use either $O(w^2)$ or $O(n)$ space. However ~\cite{GG88,UW93} reduced the space usage to $O(w)$ while maintaining the run time. A faster algorithm was given by Cole and Hariharan~\cite{CH98}, which has a runtime of $O(n(1+k^4/w))$. We refer the interested readers to a beautiful survey by Navarro~\cite{Nav01} for a comprehensive treatment on this topic. 
%In $1997$ Sahinalp and Vishkin~\cite{SV97} gave a new upper bound of $O(nk^8(\alpha \log^{\ast}n)^{1/\log 3})$. 

All the above mentioned algorithms assume that the entire text is available from the very begining of the process. However in the online version, the pattern is given at the beginning and the text arrives in a stream, one symbol at a time. Clifford {\em et al.}~\cite{CEPE08} gave a "black-box algorithm" for online approximate matching where the supported distance metrics are hamming distance, matching with wildcards, $k$-mismatch, $L_1$ and $L_2$ norm. Their algorithm has a run time of $O(\sum_{j=1}^{\log_2 w} T(n,2^{j-1})/n)$ per symbol arrival, where $T(n,w)$ is the running time of the best offline algorithm. This result was extended in ~\cite{CS09} by introducing an algorithm solving online approximate pattern matching under edit distance metric in time $O(k\log w)$ per symbol arrival. This algorithm uses $O(w)$-space. In~\cite{CS10} the runtime was further improved to $O(k)$ per symbol. However none of these algorithms for edit distance metric is black-box and they highly depend on the specific struture of the corresponding offline algorithm. Furthermore all these algorithms use linear space. Recently, Starikovskaya~\cite{STA17} gave a randomized algorithm which has a worst case time complexity of $O((k^2\sqrt{w}+k^{13})\log^4 w)$ and uses space $O(k^8\sqrt{w}\log^6 w)$. Although her algorithm takes both sublinear time and sublinear space for small values of $k$, heavy dependancy on $k$ in the complexity terms makes it much worse than the previously known algorithms in the high regime of $k$. On the lower bound side, Clifford, Jalsenius and Sach~\cite{CJS15} showed in the {\em cell-probe model} that expected amortized run time of any randomized algorithm solving online approximate pattern matching problem must be $\Omega(\sqrt{\log w}/(\log \log w)^{3/2})$ per output.

\section{Preliminaries}
\label{sec:prel}

We recall some basic definitions of \cite{CDGKS18}.
Consider the text $T$ of length $n$ to be aligned along the horizontal axis and the pattern $P$ of length $w$ to be aligned along the vertical axis. For $i\in \{1,\dots,n\}$, $T_i$ denotes the $i$-th symbol of $T$ and for $j\in \{1,\dots,w\}$, $P_j$ denotes the $j$-th symbol of $P$. $T_{s,t}$ is the substring of $T$ starting by the $s$-th symbol and ending by the $t$-th symbol of $T$. %the width of the interval $\mu(I)$ is $\max(I)-\month(I)=|I|-1$. 
For any interval $I\subseteq\{0,\dots,n\}$, $T_I$ denotes the substring of $T$ indexed by $I\setminus\{\min(I)\}$ and for $J\subseteq\{0,\dots,w\}$, $P_J$ denotes the substring of $P$ indexed by $J\setminus\{\min(J)\}$. 

\paragraph*{Edit distance and pattern matching graphs. }
For a text $T$ of length $n$ and a pattern $P$ of length $w$, the {\em edit distance graph} $G_{T,P}$ is a directed weighted graph called a grid graph with vertex set $\{0,\cdots,n\}\times\{0,\cdots,w\}$ and following three types of edges: $(i-1,j) \to (i,j)$ (H-steps), $(i,j-1) \to (i,j)$ (V-steps) and $(i-1,j-1) \to (i,j)$ (D-steps). Each H-step or V-step has cost $1$ and each D-step costs $0$ if $T_i=P_j$ and $1$ otherwise. The {\em pattern matching graph} $\tilde{G}_{T,P}$ is the same as the edit distance graph $G_{T,P}$ except for the cost of horizontal edges $(i,0)\to(i+1,0)$ which is zero. 

For $I\subseteq\{0,\dots,n\}$ and $J\subseteq\{0,\dots,w\}$, ${G}_{T,P}(I\times J)$ is the subgraph of ${G}_{T,P}$ induced on $I\times J$. Clearly, ${G}_{T,P}(I\times J)\cong {G}_{T_I,P_J}$. We define the cost of a path $\tau$ in ${G}_{T_I,P_J}$, denoted by $\cost_{{G}_{T_I,P_J}}(\tau)$, as the sum of the costs of its edges. We also define the cost of a  graph ${G}_{T_I,P_J}$, denoted by $\cost({G}_{T_I,P_J})$, as the cost of the cheapest path from $(\min I,\min J)$ to $(\max I,\max J)$. 

The following is well known in the literature (e.g. see~\cite{SEL80}).

 \begin{proposition}
 \label{prop:basicfact_pattern}
 Consider a pattern $P$ of length $w$ and a text $T$ of length $n$, and let $G=\tilde{G}_{T,P}$. 
For any $t \in \{1,\dots,n\}$, let $I=\{0,\cdots,t\}$, and $J=\{0,\cdots,w\}$. Then $k_t=\cost(G(I\times J))=\min_{i\le t} \editd(T_{i,t},P)$. 
 \end{proposition}
  A similar proposition is also true for the edit distance graph.
  \begin{proposition}
  \label{prop:basicfact_edit}
  Consider a pattern $P$ of length $w$ and a text $T$ of length $n$, and let $G=G_{T,P}$. For any $i_1\le i_2 \in \{1,\cdots,n\}$, $j_1\le j_2 \in \{1,\cdots,w\}$ let $I=\{i_1-1,\cdots,i_2\}$ and $J=\{j_1-1,\cdots,j_2\}$. Then $\cost(G(I\times J))=\editd(T_{i_1,i_2},P_{j_1,j_2})$. 
  \end{proposition}

 Let $G$ be a grid graph on $I\times J$ and $\tau=(i_1,j_1),\dots,(i_l,j_l)$ be a path in $G$. {\em Horizontal projection} of a path $\tau$ is the set $\{i_1,\dots,i_l\}$. Let $I'$ be a set contained in the horizontal projection of $\tau$, then $\tau_{I'}$ denotes the (unique) minimal subpath of $\tau$ with horizontal projection $I'$. Let $G'=G(I'\times J')$ be a subgraph of $G$. For $\delta\in [0,1]$ we say that $I'\times J'$ {\em $(1-\delta)$-covers} the path $\tau$ if the initial and the final vertex of $\tau_{I'}$ are at a vertical distance of at most $\delta (|I'|-1)$ from $(\min(I'),\min(J'))$ and $(\max(I'),\max(J'))$, resp.. 

 A {\em certified box} of $G$ is a pair $(I'\times J',\ell)$ where $I'\subseteq I$, $J'\subseteq J$ are intervals, 
 and $\ell \in \N$ such that $\cost(G(I'\times J')) \le \ell$. At high level, our goal is to approximate each path $\tau$ in $G$ by a path via the corner vertices of certified boxes. For that we want that a substantial portion of the path $\tau$ goes via those boxes and that the sum of the costs of the certified boxes is not much larger than the actual cost of the path. The next definition makes our requirements precise.
 Let $\sigma=\{(I_1\times J_1,\ell_1),(I_2\times J_2,\ell_2),\dots, (I_m\times J_m,\ell_m)\}$ be a sequence of certified boxes in $G$. 
 Let $\tau$ be a path in $G(I\times J)$ with horizontal projection $I$. 
 For any $k,\zeta \ge 0$, we say that $\sigma$ {\em $(k,\zeta)$-approximates} $\tau$ if the following three conditions hold:
 \begin{enumerate}
 \item $I_1,\dots,I_m$ is a decomposition of $I$, i.e., $I=\bigcup_{i\in [m]} I_i$, and for all $i\in[m-1]$, $\min(I_{i+1})=\max(I_{i})$.
 \item For each $i\in [m]$, $I_i\times J_i$ $(1-\ell_i/(|I_i|-1))$-covers $\tau$.
 \item $\sum_{i\in [m]}  \ell_i \le k \cdot \cost(\tau) + \zeta$.
 \end{enumerate}

\section{Offline approximate pattern matching}

To prove Theorem \ref{thm:main-offline} we design an algorithm as follows. 
For $k=2^j$, $j=0,\dots,\log w^{3/4}$, we run the standard $O(kn)$ algorithm ~\cite{GG88} to identify all $t$ such 
that $k_t \le k$. To identify positions with $k_t \le k$ for $k> w^{3/4}$ where $k$ is a power of two 
we will use the technique of \cite{CDGKS18} to compute $(O(1),O(w^{3/4}))$-approximation
of $k_1,\dots,k_{n}$. The obtained information can be combined in a straightforward manner to get a single $O(1)$-approximation 
to $k_1,\dots,k_{n}$: For each $t$, if for some $2^j \le w^{3/4}$, $k_t$ is at most $2^j$ (as determined by the former algorithm)
then output the smallest such $2^j$ as the approximation of $k_t$, otherwise output the approximation of $k_t$ found by the latter algorithm.
This way, for $k_t \le w^{3/4}$ we will get $2$-approximation, and for $k > w^{3/4}$ we will get a $O(1)$-approximation.
We will now elaborate on the latter algorithm based on \cite{CDGKS18}.
The edit distance algorithm of \cite{CDGKS18} has two phases which we will also use. The first phase ({\em covering phase}) identifies a set of {\em certified boxes}, subgraphs of the pattern matching graph
with good upper bounds on their cost. These certified boxes should cover the min-cost paths of interest. Then the next phase runs a min-cost path algorithm on these boxes to obtain the output sequence.
Both of these phases will take $\tildeO(n w^{3/4})$ time so the overall running time of our algorithm will be $\tildeO(n w^{3/4})$.

We describe the algorithms for the two phases next.
The algorithm will use the following parameters: $w_1=w^{1/4}$, $w_2=w^{1/2}$, $d=w^{1/4}$, $\theta=w^{-1/4}$.
The meaning of the parameters is essentially the same as in \cite{CDGKS18} and we will see it in a moment but their setting is different. 
Let $c_0,c_1 \ge 0$ be the large enough constants from \cite{CDGKS18}.
For simplicity
we will assume without loss of generality that $w_1$ and $w_2$ are powers of two (by rounding them down to the nearest powers of two),
$1/\theta$ is a reciproval of a power of two (by decreasing $\theta$ by at most a factor of two),
$w_2|w$ (by chopping off a small suffix from $P$ which will affect the approximation by a negligible additive error as $w^{3/4} \gg w_2$), and $w|n$ (if not we
can run the algorithm twice: on the largest prefix of $T$ of length divisible by $w$ and then on the largest suffix of $T$ of length divisible by $w$).
The algorithm will
not explicitly compute $k_t$ for all $t$ but only for $t$ where $t$ is a multiple of $w_2$, and then it will use the same value for each block
of $w_2$ consecutive $k_t$'s. Again, this will affect the approximation by a negligible additive error.

\section{Covering phase}
\label{sec:covering-phase}

We describe the first phase of the algorithm now.
First, we partition the text $T$ into substrings $T^0_1,\dots,T^0_{n_0}$ of length $w$, where $n_0 = n/w$. Then we process each of the parts independently. Let $T'$ be one of the parts.
We partition $T'$ into substrings $T^1_1,T^1_2,\dots,T^1_{n_1}$ of length $w_1$, 
and we also partition $T'$ into substrings $T^2_1,T^2_2,\dots,T^2_{n_2}$ of length $w_2$, where $n_1=w/w_1$ and $n_2=w/w_2$. 
For a substring $u$ of $v$ starting by $i$-th symbol of $v$ and ending by $j$-th symbol of $v$, we let $\{i-1,i+1,\dots,j,j\}$ be its {\em span}.
Then the covering algorithm proceeds as follows:

\paragraph*{Dense substrings.}
In this part the algorithm aims to identify for each $\epsilon_j$, that is a power of two, a set of substrings $T^1_i$ which are
similar to more than $d$ {\em relevant} substrings of $P$. (A string is relevant if it starts at a position $j$ such that $j-1$ is divisible by $\epsilon_j w_1/8$
and it is of the same length as $T^1_i$.)
We identify each $T^1_i$ by testing a random sample of relevant substrings of $P$. If we determine with high confidence that there
are at least $\Omega(d)$ substrings of $P$ similar to $T^1_i$, we add $T^1_i$ into a set $D_j$ of such strings, and we also identify
all $T^1_{i'}$ that are similar to $T^1_i$. By triangle inequality we would also expect them to be similar to many relevant substrings of 
$P$. So we add these $T^1_{i'}$ to $D_j$ as well as we will not need to process them anymore. We output the set of certified boxes
of edit distance $O(\epsilon_j w_1)$ found this way. More formally: 

For $j=\lceil \log 1/\theta \rceil,\dots,0$, the algorithm maintains sets $D_j$ of substrings $T^1_i$. These sets are initially empty. 

\smallskip\noindent {{\bf Step 1.}}
For each $i=1,\dots,n_1$ and $j=\lceil \log 1/\theta \rceil,\dots,0$, if $T^1_i$ is in $D_j$ then we continue with the next $i$ and $j$. Otherwise we process it as follows.

\smallskip\noindent {{\bf Step 2.}}
Set $\epsilon_j=2^{-j}$.
Independently at random, sample $8c_0 \cdot w \cdot (\epsilon_j w_1 d)^{-1} \cdot \log n$ many $(\epsilon_j /8)$-aligned substrings of $P$ of length $w_1$. (By an {\em $\ell$-aligned substring} of length $w_1$ in $P$ we mean a substring starting
by a symbol at a position $j$ such that $j-1$ is a multiple of $\max(\lfloor \ell w_1 \rfloor,1)$.) For each sampled substring $u$ check if its edit distance from $T^1_i$ is at most $\epsilon_j w_1$.
If less than $\frac{1}{2} \cdot c_0 \cdot \log n$ of the samples have their edit distance from $T^1_i$ below $\epsilon_j w_1$ then we are done with processing this $i$ and $j$ and we continue with the next pair. 

\smallskip\noindent {{\bf Step 3.}}
Otherwise 
%we add the span of $T^1_i$ into $D_j$ and furthermore, 
we identify all substrings $T^1_{i'}$ that are not in $D_j$ and are at edit distance at most $2\epsilon_j w_1$ from $T^1_i$, and we let $X$ to be the set of their spans relative to the whole $T$.
% We might allow also some substrings $T^1_{i'}$ of edit distance at most $4\epsilon_j w_1$ to be included in the set $X$ (as some might be misidentified to have the smaller edit distance from $T^1_i$).

\smallskip\noindent {{\bf Step 4.}}
Then we identify all  $(\epsilon_j/8)$-aligned substrings of $P$ of length $w_1$ that are are at edit distance at most $3\epsilon_j w_1$ from $T^1_i$, and we let $Y$ to be the set of their spans.
We might allow also some $(\epsilon_j w_1/8)$-aligned substrings of $P$ of edit distance at most $6\epsilon_j w_1$ to be included in the set $Y$ (as some might be misidentified to have the smaller edit distance from $T^1_i$ by our procedure that searches for them).

\smallskip\noindent {{\bf Step 5.}}
For each pair of spans $(I,J)$ from $X\times Y$ we output corresponding certified box $(I\times J,8\epsilon_j w_1)$. 
We add substrings corresponding to $X$ into $D_j$ and continue with the next pair $i$ and $j$. 

Once we process all pairs of $i$ and $j$, we proceed to the next phase: {\em extension sampling}.

\paragraph*{Extension sampling.}

In this part for every $\epsilon_j=2^{-j}$ and every substring $T^2_i$, which does not have all its substrings $T^1_\ell$ contained
in $D_j$ we randomly sample a set of such $T^1_\ell$'s. For each sampled $T^1_\ell$ we determine all relevant
substrings of $P$ at edit distance at most $\epsilon_j w_1$ from $T^1_\ell$. There should be $O(d)$-many such substrings of $P$.
We extend each such substring into a substring of size $|T^2_i|$ within $P$ and we check the edit distance of the extended string
from $T^2_i$. For each extended substring of edit distance at most $3\epsilon_j w_2$ we output a set of certified boxes.  

Here we define the appropriate extension of substrings. Let $u$ be a substring of $T$ of length less than $|P|$, and let $v$ be a substring of $u$ starting by the $i$-th symbol of $u$. 
Let $v'$ be a substring of $P$ of the same length as $v$ starting by the $j$-th symbol of $P$. The {\em diagonal extension $u'$ of $v'$ in $P$ with respect to $u$ and $v$},
is the substring of $P$ of length $|u|$ starting at position $j-i$. If $(j-i)\le 0$ then the extension $u'$ is the prefix of $P$ of length $|u|$, and if $j-i+|P| > |P|$
then the extension $u'$ is the suffix of $P$ of length $|u|$.

\smallskip\noindent {{\bf Step 6.}}
Process all pairs $i=1,\dots,n_2$ and $j=\lceil \log 1/\theta \rceil,\dots,0$. 

\smallskip\noindent {{\bf Step 7.}}
Independently at random, sample $c_1 \cdot \log^2 n \cdot \log w$ substrings $T^1_\ell$ that are part of $T^2_i$ and that are not in $D_j$. 
(If there is no such substring continue for the next pair of $i$ and $j$.)

\smallskip\noindent {{\bf Step 8.}}
For each $T^1_\ell$, find all $(\epsilon_j/8)$-aligned substrings $v'$ of $P$ of length $w_1$ 
that are at edit distance at most $\epsilon_j w_1$ from $T^1_\ell$.

\smallskip\noindent {{\bf Step 9.}}
For each $v'$ determine its diagonal extension $u'$ with respect to $T^2_i$ and $T^1_\ell$. 
Check if the edit distance of $u'$ and $T^2_i$ is less than $3\epsilon_j w_2$.
If so, compute it and denote the distance by $c$.
Let $I'$ be the span of $T^2_i$ relative to $T$, and $J'$ be the span of $u'$ in $P$. 
For all powers $a$ and $b$ of two, $w^{3/4} \le a \le b \le w$, output the certified box $(I'\times J',c+a+b)$. Proceed for the next $i$ and $j$.

This ends the covering algorithm which outputs various certified boxes.

\smallskip
To implement the above algorithm we will use Ukkonen's ~\cite{UKK85} $O(nk)$-time algorithm to check whether the edit distance of two strings of length $w_1$ is at most $\epsilon_j w_1$
in time $O(w_1^2 \epsilon_j)$. Given the edit distance is within this threshold the algorithm can also output its precise value. 
To identify all substrings of length $w_1$ at edit distance at most $\epsilon_j w_1$ of $S$ from a given string $R$ (where $S$ is the pattern $P$ of length $w$ and $R$ is one of the $T^1_i$ of length $w_1$) 
we use the $O(nk)$-time pattern matching algorithm of Galil and Giancarlo~\cite{GG88}. 
For a given threshold $k$, this algorithm determines for each position $t$ in $S$, whether there is a substring of edit distance at most $k$ from $R$ ending at that position in $S$.
If the algorithm reports such a position $t$ then we know by the following proposition that the substring $S_{t-|R|+1,t}$ is at edit distance at most $2k$. At the same time we are
guaranteed to identify all the substrings of $S$ of length $w_1$ at edit distance at most $k$ from $R$. Hence in Step 4, finding all the substrings at distance $3\epsilon_j w_1$ with perhaps some extra
substrings of edit distance at most $6\epsilon_j w_1$ can be done in time $O(w w_1 \epsilon_j)$.

\begin{proposition}\label{prop-cut}
For strings $S$ and $R$ and integers $t \in \{1,\dots,|S|\}$, $k\ge 0$ , if $\min_{i\le t} \editd(S_{i,t},R) \le k$ then $\editd(S_{t-|R|+1,t},R) \le 2k$.
\end{proposition}

\begin{proof}
Let $S_{i,t}$ be the best match for $R$ ending by the $t$-th symbol of $S$. Hence, $k = \editd(S_{i,t},R)$.
If $S_{i,t}$ is by $\ell$ symbols longer that $R$ then $k \ge \ell$ and $\editd(S_{t-|R|+1,t},R) \le k + \ell \le 2k$ by the triangle inequality.
Similarly, if $S_{i,t}$ is shorter by $\ell$ symbols.
\end{proof}

\subsection{Correctness of the covering algorithm}

\begin{lemma}
\label{lem-cov}
Let $t\ge 1$ be such that $t$ is a multiple of $w_2$. Let $\tau_t$ be the min-cost path between vertex $(t-w,0)$ and $(t,w)$
in the edit distance graph $G=G_{T,P}$ of $T$ and $P$ of cost at least $w^{3/4}\ge \theta w$. 
The covering algorithm outputs a set of weighted boxes $\mathcal{R}$ such that every $(I\times J, \ell)\in \mathcal{R}$ is correctly certified i.e., $\cost(G(I\times J))\le \ell$ and there is a subset of $\mathcal{R}$ that $(O(1),O(k_t))$-approximates $\tau_t$ with probability at least $1-1/n^7$.
\end{lemma}

It is clear from the description of the covering algorithm that it outputs only correct certified boxes from the edit distance graph of $T$ and $P$, that is for each box $(I\times J,\ell)$, $\cost(G(I\times J)) \le \ell$.

The cost of $\tau_t$ corresponds to the edit distance between $P$ and $T_{t-w+1,t}$ and it is bounded by $2 k_t$ by Proposition \ref{prop-cut}.
Let $k'_t$ be the smallest power of two $\ge k_t$. We claim that by essentially the same argument as in Proposition 3.8 and Theorem 3.9 of \cite{CDGKS18} the algorithm outputs with high probability a set of certified boxes that $(O(1),O(k'_t))$-approximates $\tau_t$.

There are differences between the current covering algorithm and that of \cite{CDGKS18}. The main substantial difference
is that the algorithm in \cite{CDGKS18} searches for certified boxes located only within $O(k_t)$ diagonals along the main diagonal
of the edit distance graph. (This rests on the observation of Ukkonen ~\cite{UKK85} that a path of cost $\le k_t$ must pass only through
vertices on those diagonals.) Here we process certified boxes in the whole matrix as each $t$ requires a different ``main'' diagonal.
Except for this difference and the order of processing various pieces the algorithms are the same.

Although technically not quite correct, one could say that the certified boxes output by the current algorithm form a superset
of boxes output by the algorithm of \cite{CDGKS18}. This is not entirely accurate as the discovery of certified boxes depends
on the number ({\em density}) of relevant substrings of $P$ similar to a given $T^1_i$. In \cite{CDGKS18} this density is measured only
in the $O(k_t)$-width strip along the main diagonal of the edit distance graphs whereas here it is measured within the whole $P$. 
(So the actual classification of substrings $T^1_i$ on {\em dense} (in $D_j$) and {\em sparse} (not in $D_j$) might differ between the two algorithms.)
However, this difference is immaterial for the correctness argument in Theorem 3.9 of \cite{CDGKS18}.

Another difference is that in Steps 4 we use $O(w w_1 \epsilon_j)$-time algorithm to search for all the similar substrings. This algorithm
will report all the substrings we were looking for and additionally it might report some substrings of up to twice the required edit distance.
This necessitates the upper bound $8\epsilon_j w_1$ in certified boxes in Step 5. It also means a loss of factor of at most two in the approximation guarantee
as the boxes of interest are reported with the cost $8\epsilon_j w_1$ instead of the 
more accurate $5\epsilon_j w_1$ of the original algorithm in \cite{CDGKS18} which would give a $(45,15 \cost(\tau_t))$-approximation.
(In that theorem $\theta w$ represents an (arbitrary) upper bound on the cost of $\tau_t$ provided it satisfies certain technical
conditions requiring that $\theta$ is large enough relative to $w$. This is satisfied by requiring that 
$\cost(\tau_t) \ge w^{3/4} \ge \theta w$.)

Another technical difference is that the path $\tau_t$ might pass through two edit distance graphs $G_{T^0_{\ell-1},P}$
and $G_{T^0_{\ell},P}$, where $t \in [(\ell-1)w+1,\ell w]$. This means that one needs to argue separately about restriction
of $\tau_t$ to $G_{T^0_{\ell-1},P}$ and $G_{T^0_{\ell},P}$. However, the proof of Theorem 3.9 in \cite{CDGKS18}
analyses approximation of the path in separate parts restricted to substrings of $T$ of size $w_2$. As both $t$ and $w$ are multiples
of $w_2$, the argument for each piece applies in our setting as well.   

\subsection{Time complexity of the covering algorithm} Now we analyze the running time:

\begin{claim}
The covering algorithm runs in time $\tO(nw^{3/4})$ with probability at least $1-1/n^{8}$.
\end{claim}

We analyse the running time of the covering algorithm for each $T'=T^0_i$ separately. We claim that the running time on $T'$ is $\tO(w^{7/4})$ so the total
running time is $\tO((n/w)w^{7/4})=\tO(nw^{3/4})$.

In Step 1, for every $i=1,\dots,n_1$ and $j=0,\dots,\log w^{1/4}$, we might sample $O(\frac{w}{\epsilon_j w_1 d} \cdot \log n)$ 
substrings of $P$ of length $w_1$ and check whether their edit distance from $T^1_i$ is at most $\epsilon_j w_1$. This takes time 
at most $\tO( \frac{w}{\epsilon_j w_1 d} \cdot \frac{w}{w_1} \cdot w_1^2 \epsilon_j) = \tO(w^2 /d) = \tO(w^{7/4})$ in total.

We say that a bad event happens either if  some substring $T^1_i$ has more than $d$ relevant substrings of $P$ having distance at most $\epsilon_j w_1$ but we sample
less than $\frac{1}{2} \cdot c_0 \log n$ of them, or if  some substring $T^1_i$ has less than $d/4$ relevant substrings of $P$ having distance at most $\epsilon_j w_1$ but we sample
more than $\frac{1}{2} \cdot c_0 \log n$ of them. By Chernoff bound, the probability of a bad event happening during the whole run of the covering algorithm is bounded
by $\exp(-O(\log n)) \le 1/n^{8}$, for sufficiently large constant $c_0$. Assuming no bad event happens we analyze the running time of the algorithm further.

Each substring $T^1_i$ that reaches Step 3 can be associated with a set of its relevant substrings in $P$ of edit distance at most $\epsilon_j w_1$ from it.
The number of these substrings is at least $d/4$ many. These substrings must be different for different strings $T^1_i$ that reach Step 3 as if they were not distinct
then the two substrings $T^1_i$ and $T^1_{i'}$ would be at edit distance at most $2 \epsilon_j w_1$ from each other and one of them would be put into $D_j$ in Step 5
while processing the other one so it could not reach Step 3. Hence, we can reach Steps 3--5 for at most $\frac{8w}{\epsilon_j w_1}\cdot \frac{4}{d}$ strings $T^1_i$.
For a given $j$ and each $T^1_i$ that reaches Step 3, the execution of Steps 3 and 4 takes $O(w w_1 \epsilon_j)$ time, hence we will spend in them $\tO(w^2/d)=\tO(w^{7/4})$ time in total.

Step 5 can report for each $j$ at most $\frac{8w}{\epsilon_j w_1}\cdot \frac{w}{w_1}$ certified boxes, so the total time spent in this step is $\tO(w^2/w_1)=\tO(w^{7/4})$ as $\epsilon_j w_1 \ge 1/4$.

Step 7 takes order less time than Step 8. In Step 8 we use Ukkonen's~\cite{UKK85} $O(nk)$-time edit distance algorithm to check the distance of strings of length $w_1$.
We need to check $\tO(n_2 \cdot \frac{w}{\epsilon_j w_1})$ pairs for the total cost $\tO(\frac{w}{w_2} \cdot \frac{w}{\epsilon_j w_1} \cdot w_1^2 \epsilon_j) = \tO(w^{7/4})$ per $j$.

As no bad event happens, for each $T^1_\ell$, there will be at most $d/4$ strings $v'$ processed in Step 9. We will spend $O(w_2^2 \epsilon_j)$ time on each of them to check for edit distance
and $O(\log^2 n)$ to output the certified boxes. Hence, for each $j$ we will spend here  $\tO(\frac{w}{w_2} \cdot d w_2^2 \epsilon_j)$ time, which is $\tO(w w_2 d)$ in total.

Thus, the total time spent by the algorithm in each of the steps is $\tO(w^{7/4})$ as required.

\section{Min-cost Path in a Grid Graph with Shortcuts}
\label{sec:short-path}

In this section we explain how we use certified boxes to calculate the approximation of $k_t$'s.
Consider any grid graph $G$. A {\em shortcut} in $G$ is an additional edge $(i,j)\to (i',j')$ with cost $\ell$, 
where $i < i'$ and $j < j'$.

Let $G_{T,P}$ be the edit distance graph for $T$ and $P$.
Let $(I\times J,\ell)$ be a certified box in $G_{T,P}$ with $|I|=|J|$. If $\ell< 1/2(|I|-1)$ add a shortcut edge $e_{I,J}$ from vertex $(\min I,\min J+\ell)$ to vertex $(\max I,\max J-\ell)$ with cost $3\ell$. Do this for all certified boxes output by the covering algorithm to obtain a graph$G'_{T,P}$. Next remove all the diagonal edges (D-steps) of cost $0$ or $1$ from graph $G'_{T,P}$ and obtain graph graph $G''_{T,P}$. 

\begin{proposition}
\label{prop-short}
If $\tau$ is a path from $(t-w,0)$ to $(t,w)$ in $G_{T,P}$ which is $(k,\zeta)$-approximated by a 
subset of certified boxes $\sigma$ by the covering algorithm then there is a path from $(t-w,0)$ to $(t,w)$ in $G''_{T,P}$ of cost at most $5\cdot (k \cdot \cost_{G_{T,P}}(\tau) + \zeta)$ consisting of shorcut edges corresponding to $\sigma$ and H and V steps.
\end{proposition}

\begin{proof}
Let $\{(I_1\times J_1,\ell_1),(I_2\times J_2,\ell_2),\dots, (I_m\times J_m,\ell_m)\}$ be the set of certified boxes that $(k,\zeta)$-approximates $\tau$ and $\sigma'$ be a subset of $\sigma$ such that for any pair $(I_r\times J_r,\ell_r)$ in $\sigma'$, $\ell_r< 1/2(|I_r|-1)$. By definition for each $i$, $\max I_i\le \min I_{i+1}$. We approximate path $\tau$ by a path $\tau'$ as illustrated in Fig.~\ref{fig:edge-set-F}(b). For each $r\in [m]$ let $p_r=(i_r,j_r)$ be the first vertex of $\tau_{I_r}$. Defime $p_{m+1}=(t,w)$. Moreover if $(I_r\times J_r,\ell_r)\in \sigma'$ then let $p_r'=(i_r',j_r')$ and $q_r'=(u_r',v_r')$ be the start and end vertex, resp., of the corresponding shortcut edge. As $I_r\times J_r$ $(1-\ell/(|I_r|-1))$-covers $\tau$, $j_r'=\min(J_r)+\ell_r\ge j_r$ and $v_r'=\max(J_r)-\ell_r\le j_{r+1}$. Hence we define $\tau'$ passing through all of the $p_r$. For each $r$ the part of $\tau'$ between vertex $p_r$ and $p_{r+1}$ can be constructed in the following way: first climb from $p_r$ to $p_r'$ using V steps, then if $(I_r\times J_r,\ell_r)\in \sigma'$ take the shorcut edge $e_{I_r,J_r}$ from $p_r'$ to $q_r'$ and then climb up to $p_{r+1}$, otherwise take H steps from $p_r'$ to reach $(i_{r+1},j_r')$ and then take V steps upto vertex $p_{r+1}$. 

\begin{center}
\begin{figure}[ht]
\centerline{\includegraphics{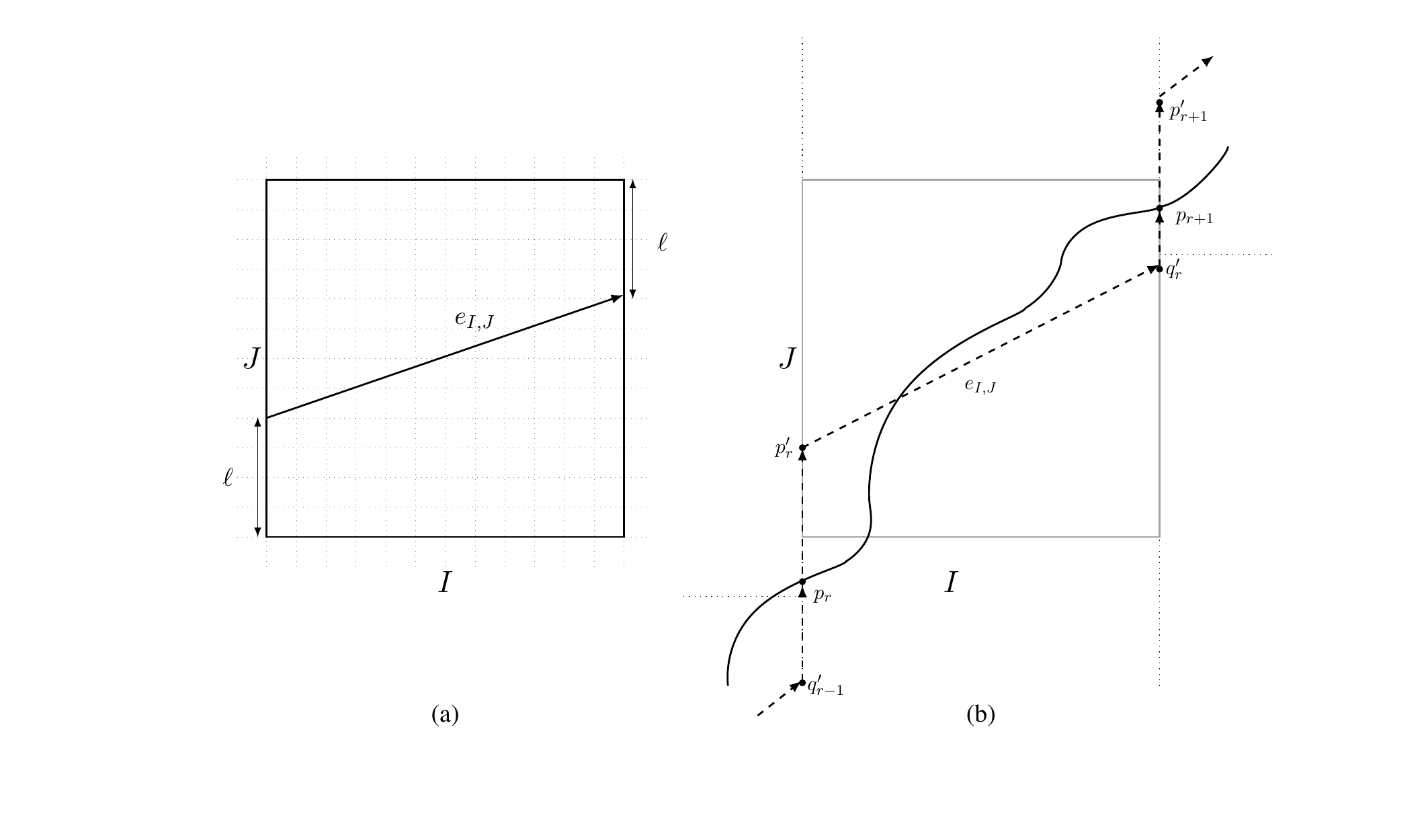}}
\caption{(a) The shortcut edge $e_{I,J}$ corresponding to a certified box $(I,J,\ell)$. (b) An example of a path $\tau$ (in solid) passing through a certified box $(I,J,\ell)$. The dashed path is an approximation $\tau'$ of $\tau$ in $G''_{T,P}$.}
   \label{fig:edge-set-F}
\end{figure}

\end{center}

Next we argue about the cost of $\tau'$. For each $r\in [m]$, if $(I_r\times J_r,\ell_r)\in \sigma'$, cost of $e_{I_r,J_r}$ is $3\ell_r$ otherwise the horizontal path with projection $I_r$ has cost $|I_r|-1\le 2\ell_r$. Hence the total cost is $\sum_r 3\ell_r$. The sum of the cost of the vertical edges is $w-\sum_{(I_r\times J_r,\ell_r)\in \sigma'}v_r'-j_r'= \sum_{(I_r\times J_r,\ell_r)\in \sigma'} (|J_r|-1)-(v_r'-j_r')+\sum_{(I_r\times J_r,\ell_r)\notin \sigma'}|J_r|-1\le \sum_r 2\ell_r$ as $\sum_r|I_r|-1=\sum_r|J_r|-1=w$. Hence the total cost of $\tau'$ is at most $\sum_r 5\ell_r$. Since $\sum_{i\in [m]}  \ell_i \le k \cdot \cost_{G_{T,P}}(\tau) + \zeta$ by definition of $(k,\zeta)$-approximation,  we get cost of $\tau'$ is at most $5\cdot (k \cdot \cost_{G_{T,P}}(\tau) + \zeta)$.      
\end{proof}

By Lemma \ref{lem-cov} and Proposition \ref{prop-short}, 
for $t$, where $w_2|t$, the cost of a shortest path from $(t-w,0)$ to $(t,w)$ in $G''_{T,P}$
is bounded by $O(k_t)$. At the same time, any path in $G''_{T,P}$ from $(i,0)$ to $(t,w)$, $i\le t$, has cost
at least $k_t$. So we only need to find the minimal cost of a shortest path from any $(i,0)$ to $(t,w)$ in $G''_{T,P}$
to get an approximation of $k_t$.

To find the minimal cost, we reset to zero the cost of all horizontal edges $(i,0)\to(i+1,0)$ in $G''_{T,P}$ to get a graph $G$. 
The graph $G$ corresponds to taking the pattern matching graph $\tilde{G}_{T,P}$, removing from it all its diagonal edges and adding the shortcut
edges. The cost of a path from $(0,0)$ to $(t,w)$ in $G$ is 
the minimum over $i\le t$ of the cost of a shortest path from $(i,0)$ to $(t,w)$ in $G''_{T,P}$.

Hence, we want to calculate the cost of the shortest path from $(0,0)$ to $(t,w)$ for all $t$.\footnote{Although, we really care only about $t$, where
$w_2 | t+1$ as for all the other values of $t$ we will approximate $k_t$ by $k_{t'}$ for the previous multiple $t'+1$ of $w_2$.} For this we will use a simple
algorithm that will make a single sweep over the shortcut edges sorted by their origin and calculate the distances for $t=0,\dots,n$.
The algorithm will maintain a data structure that at time $t$ will allow to answer efficiently queries 
about the cost of the shortest path from $(0,0)$ to $(t,j)$ for any $j \in \{0,\dots, w\}$.

The data structure will consist of a binary tree with $w+1$ leaves. Each node is associated with a subinterval of $\{0,\dots,w\}$
so that the $j$-th leaf (counting from left to right) corresponds to $\{j\}$, 
and each internal node corresponds to the union of all its children. We denote by $I_v$ the interval associated with a node $v$.
The depth of the tree is at most $1+\log (w+1)$. At time $t$, query to the node $v$ of the data structure will return
the cost of the shortest path from $(0,0)$ to $(t,\max I_v)$ that uses some shortcut edge $(i,j)\to(i',j')$, where $j' \in I_v$.
Each node $v$ of the data structure stores a pair of numbers $(c_v,t_v)$, where $c_v$ is the cost of the relevant shortest path 
from $(0,0)$ to $(t_v,\max I_v)$ and $t_v$ is the time it was updated the last time. (Initially this is set to $(\infty,0)$.)
At time $t\ge t_v$, the query to the node $v$ returns $c_v+(t-t_v)$.  

At time $t$ to find the cost of the shortest path from $(0,0)$ to $(t,j)$ we traverse the data structure from the root to the leaf $j$.
Let $v_1,\dots,v_\ell$ be the left children of the nodes along the path in which we continue to the right child. 
We query nodes $v_1,\dots,v_\ell$ to get answers $a_1,\dots,a_\ell$. The cost of the shortest paths from $(0,0)$ to $(t,j)$
is $a=\min \{j,a_1+(j-\max I_{v_1}),a_2+(j-\max I_{v_2}),\dots, a_\ell+(j-\max I_{v_\ell})\}$. As each query takes $O(1)$ time to answer, computing
the shortest path to $(t,j)$ takes $O(\log w)$ time.

The algorithm that outputs the cheapest cost of any path from $(0,0)$ to $(t,w)$ in $G$ will process the shortcut edges $(i,j)\to(i',j')$ one by one
in the order of increasing $i$. The algorithm will maintain lists $L_0,\dots,L_n$ of updates to the data structure to be made
before time $t$. At time $t$ the algorithm first outputs the cost of the shortest path from $(0,0)$ to $(t,w)$. 
Then it takes each shortcut edge $(t,j)\to(t',j')$ one by one, $t<t'$. (The algorithm ignores shortcut edges where $t=t'$.) 
Using the current state of the
data structure it calculates the cost $c$ of a shortest path from $(0,0)$ to $(t,j)$ and adds $(c+d,j')$ to list $L_{t'}$,
where $d$ is the cost of the shortcut edge $(t,j)\to(t',j')$. 

After processing all edges starting at $(t,\cdot)$
the algorithm performs updates to the data structure according to the list $L_{t+1}$.
Update $(c,j)$ consists of traversing the tree from the root to the leaf $j$ and in each node $v$ updating its current
values $(c_v,t_v)$ to the new values $(c'_v,t+1)$, where $c'_v = \min \{c_v + t + 1 - t_v, c + \max I_v - j\}$.

Then the algorithm increments $t$ and continues with further edges.

If the number of shortcut edges is $m$ then the algorithm runs in time $O(n+ m (\log m + \log w))$. First, it has to set-up the
data structure, sort the edges by their origin and then it processes each edge. Processing each edge will require $O(\log w)$ time
to find the min-cost path to the originating vertex and then later at time $t'$ it will require time $O(\log w)$ to 
update the data structure. As there are $\tO(\frac{n}{w} \cdot \frac{w}{\theta w_1} \cdot \frac{w}{w_1}) \le \tO(nw^{3/4})$ certified boxes in total the running time of the algorithm is as required.

The correctness of the algorithm is immediate from its description.

\section{Online approximate pattern matching}
\label{sec:online-pattern-matching}
In this section we describe the online algorithm from Theorem~\ref{thm:main-online}. It is based on interleaved execution of the cover and min-cost path algorithms from Sections \ref{sec:covering-phase} and \ref{sec:short-path} where we also need to maintain some extra datastructure in a clever manner for the covering algorithm. Also to get the required space bound we use a little modified tree data structure for the min-cost path algorithm. 
We will use the same parameters as there but we will set their values slightly differently:   $w_1=w^{11/18}$, $w_2=w^{20/27}$, $d=w^{7/54}$, $\theta=w^{-1/9}$.

We explain now how to interleave the two algorithms to achieve required time and space bound. 
For each substring $T^0_m$ of $w$ consecutive input symbols, and $j=\lceil \log 1/\theta \rceil,\dots,0$ the algorithm will maintain a set $D'_j$ of the content of strings $T^1_i$ that reached
Step 3 of the covering algorithm during processing of $T^0_m$ and for each of strings it will also store a set $Y_{i,j}$ of spans obtained in Step 4. This is done as we the whole $w$ length string $T^0_m$ can't be stored at once. Moreover to bound the size of $D'_j$ and $Y_{i,j}$ before adding a new $T^1_i$ that reached
Step 3 of the covering algorithm to $D'_j$, we first ensure that no string close to $T^1_i$ is already contained in $D'_j$. After finishing each $T^0_m$ we discard all this information.

The algorithm processes the input text $T$ in batches
of $w_2$ symbols. Upon receipt of the $t$-th symbol we buffer the symbol, if $t$ is not divisible by $w_2$ then the algorithm outputs the previous value $k_{t-1}$ as the current value $k_t$
and waits for the next symbol.
Otherwise we received batch $T^2_\ell$ of next $w_2$ symbols, for $\ell=t/w_2$, and we will proceed as follows.

We will execute the covering algorithm twice on $T^2_\ell$ where during the first execution the only thing that we will send to the min-cost path algorithm are the certified boxes produced at Step 9,
all other modifications to data structures will be discarded. During the second run of the algorithm on $T^2_\ell$, we will preserve all modifications to $D'_j$'s and other data structures
except we will discard the certified boxes produced at Step 9 (we will not send them to the min-cost path algorithm as they already got there in the first pass).

We will maintain sets $S_j$, $j=\lceil \log 1/\theta \rceil,\dots,0$. We empty all of them at this point.
We partition $T^2_\ell$ into $T^1_g,\dots,T^1_h$ of length $w_1$, where $g=(\ell-1)\cdot \frac{w_2}{w_1}+1$ and $h=g+\frac{w_2}{w_1}-1$. For $i=g,\dots,h$ we do the following.
For each $j=\lceil \log 1/\theta \rceil,\dots,0$, set $\epsilon_j = 2^{-j}$. Check, whether $T^1_j$ is at edit distance at most $2\epsilon_j w_1$ from some string $T^1_{i'}$ in $D'_j$.
If it is then send the set of all the certified boxes $(I,J,8\epsilon_jw_1)$ to the min-cost path algorithm, where $I$ is the span of $T^1_i$ in $T$ and $J\in Y_{i',j}$.
If it is not close to any string in $D'_j$ then sample the relevant substring in $P$ as in Step 2 and see how many of them are at edit distance $\le \epsilon_j w_1$ from $T^1_i$.
If at most $\frac{1}{2} \cdot c_0 \cdot \log n$ of the samples have their edit distance from $T^1_i$ below $\epsilon_j w_1$ then put index $i$ into $S_j$ and continue for another $j$ and then the next $i$.
Otherwise we execute Step 4 of the algorithm to find set $Y$. (We always skip Step 3.) We put $T^1_i$ into $D_j$ and set $Y_{i,j}$ to $Y$.
During the second pass over the algorithm, we send all the certified boxes $(I,J,8\epsilon_j w_1)$ to the min-cost path algorithm, 
where $I$ is the span of $T^1_i$ and $J \in Y_{i,j}$. Upon processing all $j$ and $i$ we continue to the sparse extension sampling part. 

For each $j=\lceil \log 1/\theta \rceil,\dots,0$, we sample from the set $S_j$ the strings $T^1_\ell$ in Step 7, and we proceed for them as in Steps 8--9.
During the first pass over the algorithm, for each certified box $(I,J,\ell')$ produced in Step 9 round up $\ell'$ to the nearest larger or equal power of two and send the box the the min-cost path algorithm.

The min-cost path algorithm receives certified boxes from the covering algorithm and it converts them into corresponding shortcut edges. The algorithm
receives the edges at two distinct phases. Edges received during the first phase corresponding to boxes that were produced at Step 9 are sorted by their 
originating vertex, stored, and processed at appropriate time steps during the next phase.
During the next phase the algorithm receives boxes $(I,J,8\epsilon_j w_1)$, where $I$ is the span of some $T^1_i$ and $J \in Y_{i,j}$. It converts them into edges and upon receiving all the edges
for a particular $T^1_i$, it sorts them according to their originating vertex. Then the min-cost path algorithm proceeds for times steps $(i-1) \cdot w_1$ to $i \cdot w_1-1$,
and processes all stored edges that originate in these time steps. During these time steps it also updates its tree data structure as in the offline case. Again we use lists for storing pending updates.
At any moment of time, the number of unprocessed edges and updates is bounded by the number of edges produced in Step 9 and edges produced for a particular string $T^1_i$.
This is at most $\tO(\frac{w}{\theta w_1})$.

Here we describe the modified tree datastructure used for the min-cost path algorithm. We round up all the edit distance estimates to powers of two. Moreover every shortcut edge corresponds to some certified box, hence the number of distinct vertical positions where the shortcut edges might originate from or lead to is bounded
by $q=\frac{8w}{\theta w_1} \cdot \log 1/\theta$. Thus the tree data structure of the min-cost path algorithm will ever perform updates to at most $q \log w$ 
distinct nodes. We do not need to store the nodes that are never updated, so the tree data structure will occupy only space $\tO(\frac{w}{\theta w_1})$.
We conclude by the following lemma:

\begin{lemma}
 \label{prop:online_complexity}
Let $n$ and $w$ be large enough integers. Let $P$ be the pattern of length $w$, $T$ be the text of length $n$ (arriving online one symbol at a time), $1/w\le \theta \le 1$ be a real. Let $\theta w_1\ge 1$, $w_1\le \theta w_2$, $w_1 | w_2$ and $w_2 |n$. With probability at least $1-1/\poly(n)$ the online algorithm for pattern matching runs in amortized time $\tO(\frac{w}{d}+\frac{ww_1}{w_2}+dw_2+\frac{w}{w_1})$ per symbol and in succinct space $\tO(w_2 + \frac{w}{d\theta}+\frac{w}{w_1\theta}+\frac{w^2}{{\theta}^2{w_1^2}d}+d)$.  
 \end{lemma}
 \begin{proof}
 The running time of the online algorithm can be analysed in a similar manner as the offline algorithm. The only difference is that here the covering and the min-cost path algorithm is interleaved. For each batch of $w_2$ symbols we run the covering algorithm twice with the modification that instead of executing Step 3 for each $T^1_j$ we check whether it is at distance at most $2{\epsilon}_jw_1$ from some string in ${D'}_j$. But this step takes amortized time $\tO(\frac{w}{{\epsilon}_jw_1d}\cdot{w_1^2}\epsilon_j\cdot \frac{1}{w_1})=\tO(\frac{w}{d})$. Moreover the total number of certified boxes send by the covering algorithm to the min-cost path algorithm is the same in both the offline and the online algorithm. Hence the online algorithm has the amortized time of $\tO(\frac{w}{d}+\frac{ww_1}{w_2}+dw_2+\frac{w}{w_1})$ per symbol.
 
 To determine the space complexity of the online algorithm we analyse the space used by the covering algorithm and the min-cost path algorithm separately. At any time the covering algorithm stores a batch of $w_2$ symbols which takes space $O(w_2)$. Next for $j=\lceil \log 1/\theta\rceil,\dots,0$ it stores set $D'_j$ of strings $T^1_i$ that reached Step 3 of the covering algorithm. Each of these strings is of length $w_1$, hence requires $O(w_1)$ space. Moreover for each such string the algorithm stores set $Y_{i,j}$ of spans obtained at Step 4 and this require space $O(\frac{w}{\epsilon_jw_1})$. For each such string (as it reached Step 3), there exist at least $d/4$ relevant substrings of $P$ which are at distance at most $\epsilon_jw_1$, and for any two strings of $D'_j$ (as they are at distance more than $2\epsilon_jw_1$) these sets of relevant substrings of $P$ are disjoint. Hence $D'_j$ stores contents of at most $\frac{4w}{\epsilon_jw_1d}$ different strings and the total space used by all $D'_j$ and $Y_{i,j}$ is $\tO(\frac{w}{\theta w_1d}\cdot w_1+\frac{w}{\theta w_1d}\cdot \frac{w}{\theta w_1})=\tO(\frac{w}{\theta d}+\frac{w^2}{\theta^2{w_1^2}d})$. Maintaining sets $S_j$ does not require any extra space as we store the whole batch of $w_2$ symbols. As argued before the tree data structure stored by the min-cost path algorithm occupies space $\tO(\frac{w}{\theta w_1})$ and the list of edges can be stored in $\tO(\frac{w}{\theta w_1}+d)$ space. Hence total succinct space used by the online algorithm is $\tO(w_2 + \frac{w}{d\theta}+\frac{w}{w_1\theta}+\frac{w^2}{{\theta}^2{w_1^2}d}+d)$. 
\end{proof}

For example, we can instantiate the above proposition for the parameters: $w_1=w^{11/18}$, $w_2=w^{20/27}$, $d=w^{7/54}$, $\theta=w^{-1/9}$, to get the following:

\begin{theorem}
%\label{thm:main-online}
There is a constant $c\ge 1$ so that there is a randomized online algorithm that computes $(c,w^{8/9})$-approximation to approximate pattern matching in amortized time $O(w^{1-(7/54)})$ and succinct space $O(w^{1-(1/54)})$ with probability at least $1-1/\poly(n)$.
\end{theorem}

\bibliography{apm}
\end{document}